\documentclass[a4paper,conference]{IEEEtran}
\IEEEoverridecommandlockouts                              

\usepackage{cite}
\usepackage{psfrag,graphicx}
\usepackage[tight,footnotesize]{subfigure}
\usepackage[cmex10]{amsmath}
\usepackage{amsthm}
\usepackage{bbm}
\usepackage{vmr-symbols-vecbold}
\usepackage{standard-macros}

\definecolor{Red}{rgb}{1.0, 0, 0}

\newtheorem{lemma}{Lemma}
\newtheorem{Propos}{Proposition}
\newtheorem{Conject}{Conjecture}

\newtheorem{theorem}{Theorem}

\begin{document}

\title{A Lower Bound on the Entropy Rate for\\a Large Class of Stationary Processes and\\its Relation to the Hyperplane Conjecture}
\author{%
\IEEEauthorblockN{Meik D\"orpinghaus}
\IEEEauthorblockA{Vodafone Chair Mobile Communications Systems\\Technische Universit\"at Dresden\\ 01062 Dresden, Germany\\ meik.doerpinghaus@tu-dresden.de}

\thanks{This work was supported in part by the Deutsche Forschungsgemeinschaft (DFG) in the framework of the cluster of excellence EXC 1056 `Center for Advancing Electronics Dresden (cfaed)' and by the DFG within the Collaborative Research Center SFB~912 `Highly Adaptive Energy-Efficient Computing (HAEC)'.}
\thanks{During the course of this work the author was a visiting assistant professor at the Department of Electrical Engineering, Stanford University, CA, U.S.A.}}

\maketitle
\thispagestyle{empty}
\pagestyle{empty}

\begin{abstract}
We present a new lower bound on the differential entropy rate of stationary processes whose sequences of probability density functions fulfill certain regularity conditions. This bound is obtained by showing that the gap between the differential entropy rate of such a process and the differential entropy rate of a Gaussian process with the same autocovariance function is bounded. This result is based on a recent result on bounding the Kullback-Leibler divergence by the Wasserstein distance given by Polyanskiy and Wu. Moreover, it is related to the famous \emph{hyperplane conjecture}, also known as \emph{slicing problem}, in convex geometry originally stated by J. Bourgain. Based on an entropic formulation of the hyperplane conjecture given by Bobkov and Madiman we discuss the relation of our result to the hyperplane conjecture.\looseness-1  
\end{abstract}

\IEEEpeerreviewmaketitle

\section{Introduction}\label{Section_Introduction}
We consider a discrete-time stationary stochastic process $\{\mathsf{X}\}$. A finite sequence of this process is given by the vector $\mathsf{X}^{n}=[\mathsf{X}_1,\hdots,\mathsf{X}_{n}]$ in $\mathbb{R}^{n}$ with density $f^{(n)}$. The differential entropy rate of the process $\{\mathsf{X}\}$ is defined as
\begin{IEEEeqnarray}{rCL}
h'(\{\mathsf{X}\})&=&\lim_{n\rightarrow\infty}\frac{1}{n}h(\mathsf{X}^n)
\end{IEEEeqnarray}
when the limit exists and where $h(\cdot)$ denotes the differential entropy.

In general, no closed form solutions for the entropy rate are available. However, for a stationary Gaussian random process $\{\mathsf{X}_{\textrm{Gauss}}\}$ having the same autocovariance function and the same mean as the process $\{\mathsf{X}\}$ the entropy rate can be expressed in closed form.
Let 
\begin{IEEEeqnarray}{rCL}
r(l)&=&\mathrm{E}[(\mathsf{X}_{k+l}-\mathrm{E}[\mathsf{X}_{k+l}])(\mathsf{X}_k-\mathrm{E}[\mathsf{X}_k])]\label{DefAutocorr}
\end{IEEEeqnarray}
be the autocovariance function of the process $\{\mathsf{X}\}$. For a blocklength of $n$ the elements of its covariance matrix of size $n\times n$ are given by 
\begin{IEEEeqnarray}{rCL}
[\mathbf{R}^{(n)}_{\mathsf{X}}]_{kl}&=&r(k-l).\label{DEFCovMat}
\end{IEEEeqnarray}
Then the entropy rate of $\{\mathsf{X}_{\textrm{Gauss}}\}$ is given by, see, e.g., \cite[Ch.~12.5]{CoverBook2}\footnote{Unless otherwise stated, in this paper $\log$ is to an arbitrary but fixed base.}
\begin{IEEEeqnarray}{rCL}
h'(\{\mathsf{X}_{\textrm{Gauss}}\})&=&\lim_{n\rightarrow\infty}\frac{1}{2n}\log\big((2\pi e)^n\det\big(\mathbf{R}^{(n)}_{\mathsf{X}}\big)\big)\nonumber\\
&=&\frac{1}{2}\int_{-\frac{1}{2}}^{\frac{1}{2}}\log\left(2\pi e S_{\mathsf{X}}(f)\right)df\label{EntRateSpec}
\end{IEEEeqnarray}
where (\ref{EntRateSpec}) follows with Szeg\"o's theorem on the asymptotic eigenvalue distribution of Hermitian Toeplitz matrices \cite[pp.~64-65]{Szgo58},\cite{Gray_ToeplitzReview} with the power spectral density $S_{\mathsf{X}}(f)$ given by\looseness-1
\begin{IEEEeqnarray}{rCL}
S_{\mathsf{X}}(f)&=&\sum_{m=-\infty}^{\infty}r(m)e^{-j 2\pi m f}, \quad -\frac{1}{2}<f < \frac{1}{2}\label{Def_PSD}
\end{IEEEeqnarray}
with $j=\sqrt{-1}$.

As for a given covariance matrix Gaussian random variables maximize the differential entropy \cite[Th.~8.6.5]{CoverBook2}, it follows that the entropy rate $h'(\{\mathsf{X}\})$ is upper-bounded by 
\begin{IEEEeqnarray}{rCL}
h'(\{\mathsf{X}\})&\le&h'(\{\mathsf{X}_{\textrm{Gauss}}\}).\label{NewGaussUpperBound}
\end{IEEEeqnarray}

While for given second moments an upper bound on the differential entropy and the corresponding rate can be easily given by considering the maximum entropy property of the Gaussian distribution as in (\ref{NewGaussUpperBound}), there exist not many general lower bounds on differential entropy and the corresponding rate. One approach of lower-bounding the differential entropy has been presented in \cite[Theorem~I.1]{BobkovMadiman11} where a lower bound on the entropy per coordinate is constructed based on a Gaussian density having the same maximum density as the actual density of the considered random vector. I.e., differently to matching the first two moments when using the maximum entropy property of the Gaussian distribution for upper-bounding, here the maximum of the density between the actual distribution and a Gaussian distribution is matched. Related to this, in \cite[Corollary~IX.1]{BobkovMadiman11} under specific conditions the entropy rate of a stationary process can be bounded away from $-\infty$. 

Differently, the lower bound on the entropy rate presented in this paper is based on an upper bound on the Kullback-Leibler divergence between the actual density and a Gaussian density with the same first two moments. In this regard, note that\looseness-1
\begin{IEEEeqnarray}{rCL}
h'(\{\mathsf{X}_{\textrm{Gauss}}\})-h'(\{\mathsf{X}\})&=&
\lim_{n\rightarrow\infty}\frac{1}{n}\left[h(\mathsf{X}_{\textrm{Gauss}}^{n})-h(\mathsf{X}^{n})\right]\nonumber\\
&=&\lim_{n\rightarrow\infty}\frac{1}{n}D(f^{(n)})\label{NewLimDiffKL}
\end{IEEEeqnarray}
where $\mathsf{X}_{\textrm{Gauss}}^{n}=[\mathsf{X}_{\textrm{Gauss},1},\hdots,\mathsf{X}_{\textrm{Gauss},n}]$ is the vector containing the elements $1,\hdots,n$ of the process $\{\mathsf{X}_{\textrm{Gauss}}\}$. Moreover, $D(f^{(n)})$ is the Kullback-Leibler divergence or relative entropy between the density $f^{(n)}$ and a corresponding Gaussian distribution with the same mean and the same covariance matrix as $\mathsf{X}^{n}\sim f^{(n)}$ and with density $g^{(n)}$, i.e.
\begin{IEEEeqnarray}{rCL}
D(f^{(n)})&=&D(f^{(n)}\Vert g^{(n)})=\int f^{(n)}(x^{n})\log \frac{f^{(n)}(x^{n})}{g^{(n)}(x^{n})}dx^{n}\nonumber\\
\end{IEEEeqnarray}
where the superscript at $x^{n}$ denotes a vector of dimension $n$. For (\ref{NewLimDiffKL}) we have used that, see, e.g., \cite[pp.~254-255]{CoverBook2} 
\begin{IEEEeqnarray}{rCL}
D(f^{(n)})&=&h(\mathsf{X}_{\textrm{Gauss}}^{n})-h(\mathsf{X}^{n})\label{KL_GaussEntropDiff}
\end{IEEEeqnarray}
as $\mathsf{X}_{\textrm{Gauss}}^{n}$ is Gaussian with the same mean and the same covariance matrix as $\mathsf{X}^{n}$.

In the present paper, we show that 
\begin{IEEEeqnarray}{rCL}
\frac{1}{n}D(f^{(n)})&\le&c<\infty\label{NewMainStatement}
\end{IEEEeqnarray}
for all $n$ if the sequence of densities $f^{(n)}$ fulfills certain regularity conditions and, hence, that
\begin{IEEEeqnarray}{rCL}
\lim_{n\rightarrow\infty}\frac{1}{n}D(f^{(n)})&\le&c<\infty.\label{NewMainStatementLim}
\end{IEEEeqnarray}
Here, $c$ is a uniform constant independent of $n$. 

The proof of (\ref{NewMainStatement}) is based on a recent result of bounding the Kullback-Leibler divergence by the Wasserstein distance given by Polyanskiy and Wu \cite{PolyanskiyWu2015}. More specifically, we show that the LHS of (\ref{NewMainStatement}) is bounded away from infinity for all $n$ if the sequence of probability density functions $f^{(n)}$ is $(c_1,c_2)$-regular with respect to the following definition given in \cite{PolyanskiyWu2015}.

A probability density function $f^{(n)}$ on $\mathbb{R}^{n}$ is $(c_1,c_2)$-regular if $c_1>0$, $c_2\ge 0$, and
\begin{IEEEeqnarray}{rCL}
\Vert\nabla\log f^{(n)}(x^{n})\Vert&\le& c_1 \Vert x^{n}\Vert +c_2, \qquad \forall x^{n}\in\mathbb{R}^{n}\label{RegularityCond}
\end{IEEEeqnarray}
where $\Vert\cdot\Vert$ denotes the Euclidean distance and with the differential operator $\nabla=\left(\frac{\partial}{\partial x_1},\hdots,\frac{\partial}{\partial x_n}\right)$ for $x^{n}=[x_1,\hdots,x_n]$.  

We state that a sequence of densities $f^{(n)}$ or a process $\{\mathsf{X}\}$ are $(c_1,c_2)$-regular if there exist constants $c_1$ and $c_2$ such that (\ref{RegularityCond}) holds for all $n\in\mathbb{N}$.

In conclusion, we state the following theorem.
\begin{theorem}\label{TheoremBoundStationary}
Let $\{\mathsf{X}\}$ be a stationary stochastic process with autocovariance function $r(l)$ and probability density function $f^{(n)}$ for a sequence of length $n$. Moreover, let $\{\mathsf{X}_{\textrm{Gauss}}\}$ be a Gaussian process with the same mean and the same autocovariance function as the process $\{\mathsf{X}\}$. If the sequence of densities $f^{(n)}$ is $(c_1,c_2)$-regular, it holds that 
\begin{IEEEeqnarray}{rCL}
\frac{1}{n}D(f^{(n)})&=&\frac{1}{n}\left[h(\mathsf{X}_{\textrm{Gauss}}^{n})-h(\mathsf{X}^{n})\right]\le c<\infty,\ \forall n\label{NewTheorem_FiniteDiff}
\end{IEEEeqnarray}
where $c$ is a uniform constant independent of $n$.
\end{theorem}

With Theorem~\ref{TheoremBoundStationary} it follows that
\begin{IEEEeqnarray}{rCL}
h'(\{\mathsf{X}_{\textrm{Gauss}}\})-h'(\{\mathsf{X}\})&=&\lim_{n\rightarrow\infty}\frac{1}{n}\left[h(\mathsf{X}_{\textrm{Gauss}}^{n})-h(\mathsf{X}^{n})\right]\le c\nonumber\\
\end{IEEEeqnarray}
and we get the following lower bound on the differential entropy rate of $\{\mathsf{X}\}$
\begin{IEEEeqnarray}{rCL}
h'(\{\mathsf{X}\})&\ge h'(\{\mathsf{X}_{\textrm{Gauss}}\})- c.
\end{IEEEeqnarray}
In Section~\ref{Section_LowerBound} we give the proof of Theorem~\ref{TheoremBoundStationary} and provide an example on the application of this lower bound.

The result stated in Theorem~\ref{TheoremBoundStationary} is related to the so called \emph{hyperplane conjecture}, also referred to as \emph{slicing problem}, which is a longstanding open problem in convex geometry. The hyperplane conjecture has been initially formulated by J. Bourgain \cite{Bourgain86} and has raised a lot of attention in almost the last 30 years. Its initial version by Bourgain was a slight variant of the following form \cite[Conj.~V.2]{BobkovMadiman11}.

\begin{Conject}[\emph{Hyperplane Conjecture} or \emph{Slicing Problem}]\label{SliceConject}
There exists a universal, positive constant $c$ (not depending on $n$) such that for any convex set $K$ of unit volume in $\mathbb{R}^{n}$, there exists a hyperplane $H$ such that the $(n-1)$-dimensional volume of the section $K\cap H$ is bounded below by $c$.
\end{Conject}

There exist several equivalent formulations of the hyperplane conjecture having a geometric or a functional analytic flavor \cite{BobkovMadiman11}. As summarized in \cite{BobkovMadiman11}, Milman and Pajor \cite{MilmanPajor} looked like Bourgain \cite{Bourgain86} on a setting of centrally symmetric, convex bodies, while Ball's formulation of the conjecture \cite{Ball} states that the isotropic constant of a log-concave measure in any Euclidean space is bounded above by a universal constant independent of the dimension. In addition to these formulations, Bobkov and Madiman gave the following entropic form of the hyperplane conjecture.

\begin{Conject}[Entropic Form of the Hyperplane Conjecture]\cite[Conjecture V.4]{BobkovMadiman11}\label{SlicingProb_Entrop}
For any log-concave density $f^{(n)}$ on $\mathbb{R}^{n}$ and some universal constant $c$
\begin{IEEEeqnarray}{rCL}
\frac{D(f^{(n)})}{n}&\le& c.\label{MainIneqConjHyp}
\end{IEEEeqnarray}
\end{Conject}

As stated in \cite{BobkovMadiman11} Conjecture~\ref{SlicingProb_Entrop} gives a formulation of the hyperplane conjecture as a statement about the (dimension-free) closeness of a log-concave measure to a Gaussian measure.

While the proof of the hyperplane conjecture is in general open, partial results are known, see, e.g., \cite{Bourgain86,MilmanPajor,klartag2009hyperplane} and references therein. 

Obviously, Conjecture~\ref{SlicingProb_Entrop} and Theorem~\ref{TheoremBoundStationary} have some similarities as they both upper-bound $D(f^{(n)})/n$. However, there are also significant differences. We will discuss the relation between both in Section~\ref{Section_Hyperplane}.

\section{Lower-Bounding of Entropy Rates}\label{Section_LowerBound}
In the following, we present a proof of Theorem~\ref{TheoremBoundStationary} based on a result by Polyanskiy and Wu given in \cite[Proposition~1]{PolyanskiyWu2015}. Before stating this proposition we recall the definition of the Wasserstein distance as given in \cite{PolyanskiyWu2015}.

The Wasserstein distance on the Euclidean space is defined as follows. Let $\mathsf{X}$ and $\mathsf{Y}$ be random vectors in $\mathbb{R}^{n}$. Given probability measures $\mu,\nu$ on $\mathbb{R}^{n}$, their $p$-Wasserstein distance ($p\ge 1$) is given by
\begin{IEEEeqnarray}{rCL}
W_p(\mu,\nu)&=&\inf\left(\mathrm{E}[\Vert \mathsf{X}-\mathsf{Y}\Vert^{p}]\right)^{1/p}\label{DefWasserstein}
\end{IEEEeqnarray}
where the infimum is taken over all couplings of $\mu$ and $\nu$, i.e., joint distributions $P_{\mathsf{XY}}$ whose marginals satisfy $P_{\mathsf{X}}=\mu$ and $P_{\mathsf{Y}}=\nu$.

Now we are ready to state the proposition given by Polyanskiy and Wu.

\begin{Propos}\cite[Prop.~1]{PolyanskiyWu2015}\label{Propo2}
Let $\mathsf{U}$ and $\mathsf{V}$ be random vectors with finite second moments. If $\mathsf{V}$ has a $(c_1,c_2)$-regular density $p_{\mathsf{V}}$, then there exists a coupling $P_{\mathsf{UV}}$, such that\footnote{Note, on the LHS of (\ref{Prop2MainIneq}) the expectation is taken with respect to the random variable $(\mathsf{U},\mathsf{V})$, i.e., 
\begin{IEEEeqnarray}{rCL}
\mathrm{E}\left[\left|\log\frac{p_{\mathsf{V}}(\mathsf{V})}{p_{\mathsf{V}}(\mathsf{U})}\right|\right]
&=&\int\int p_{\mathsf{U}\mathsf{V}}(u,v)\left|\log\frac{p_{\mathsf{V}}(v)}{p_{\mathsf{V}}(u)}\right|dudv\nonumber
\end{IEEEeqnarray}
where $p_{\mathsf{U}\mathsf{V}}$ is the joint density of $(\mathsf{U},\mathsf{V})$.}
\begin{IEEEeqnarray}{rCL}
\mathrm{E}\left[\left|\log\frac{p_{\mathsf{V}}(\mathsf{V})}{p_{\mathsf{V}}(\mathsf{U})}\right|\right]&\le& \Delta\label{Prop2MainIneq}
\end{IEEEeqnarray}
where 
\begin{IEEEeqnarray}{rCL}
\Delta&=&\left(\frac{c_1}{2}\sqrt{\mathrm{E}\left[\left\Vert\mathsf{U}\right\Vert^{2}\right]}+\frac{c_1}{2}\sqrt{\mathrm{E}\left[\left\Vert\mathsf{V}\right\Vert^{2}\right]}+c_2\right)W_2(P_{\mathsf{U}},P_{\mathsf{V}}).\nonumber\\\label{DefDelta}
\end{IEEEeqnarray}
Consequently,
\begin{IEEEeqnarray}{rCL}
h(\mathsf{U})-h(\mathsf{V})&\le&\Delta.\label{UsedInequality}
\end{IEEEeqnarray}
If both $\mathsf{U}$ and $\mathsf{V}$ are $(c_1,c_2)$-regular, then
\begin{IEEEeqnarray}{rCL}
|h(\mathsf{U})-h(\mathsf{V})|&\le& \Delta\\
D(P_{\mathsf{U}}\Vert P_{\mathsf{V}})+D(P_{\mathsf{V}}\Vert P_{\mathsf{U}})&\le& 2\Delta.
\end{IEEEeqnarray}
\end{Propos}

Based on Proposition~\ref{Propo2} we now prove Theorem~\ref{TheoremBoundStationary}.

\begin{proof}[Proof of Theorem~\ref{TheoremBoundStationary}]
We assume that the sequence of densities $f^{(n)}$ is $(c_1,c_2)$-regular. Hence, there exist a finite $c_1>0$ and a finite $c_2\ge 0$ such that (\ref{RegularityCond}) holds for all $n\in \mathbb{N}$. We further assume that $g^{(n)}$ is a Gaussian density having the same mean and the same covariance matrix as $f^{(n)}$. As $g^{(n)}$ is Gaussian the Kullback-Leibler divergence between $f^{(n)}$ and $g^{(n)}$ can be expressed as a difference of differential entropies, see (\ref{KL_GaussEntropDiff})
\begin{IEEEeqnarray}{rCL}
D(f^{(n)})=D(f^{(n)}\Vert g^{(n)})&=&h(g^{(n)})-h(f^{(n)})\nonumber\\
&=&h(\mathsf{X}_{\textrm{Gauss}}^{n})-h(\mathsf{X}^{n})
\end{IEEEeqnarray}
where $\mathsf{X}^{n}\sim f^{(n)}$ and $\mathsf{X}_{\textrm{Gauss}}^{n}\sim g^{(n)}$. The Kullback-Leibler divergence is always nonnegative \cite[Theorem~8.6.1]{CoverBook2}.  

Identifying $\mathsf{U}$ and $\mathsf{V}$ in Prop.~\ref{Propo2} with $\mathsf{X}_{\textrm{Gauss}}^{n}$ and $\mathsf{X}^{n}$, respectively, with (\ref{UsedInequality}) and (\ref{DefDelta}) $D(f^{(n)})$ is  upper-bounded by
\begin{IEEEeqnarray}{rCL}
D(f^{(n)})&\le& \left(\frac{c_1}{2}\sqrt{\mathrm{E}\left[\left\Vert\mathsf{X}_{\textrm{Gauss}}^{n}\right\Vert^{2}\right]}+\frac{c_1}{2}\sqrt{\mathrm{E}\left[\left\Vert\mathsf{X}^{n}\right\Vert^{2}\right]}+c_2\right)\nonumber\\
&&\times W_2(g^{(n)},f^{(n)})\nonumber\\
&=&\left(c_1\sqrt{\mathrm{E}\left[\left\Vert\mathsf{X}^{n}\right\Vert^{2}\right]}+c_2\right)\inf\sqrt{\mathrm{E}\left[\Vert\mathsf{X}_{\textrm{Gauss}}^{n}-\mathsf{X}^{n}\Vert^{2}\right]}\nonumber\\
\label{ProofLast-2}
\end{IEEEeqnarray}
where for (\ref{ProofLast-2}) we have substituted the Wasserstein distance $W_2(g^{(n)},f^{(n)})$ using its definition in (\ref{DefWasserstein}). The infimum is taken over all couplings of $g^{(n)}$ and $f^{(n)}$. Note, as we assume that the densities $f^{(n)}$ and $g^{(n)}$ exist, we have substituted the distributions in the argument of Wasserstein distance in (\ref{DefDelta}) by the corresponding densities. Moreover, we have used that $\mathsf{X}^{n}$ and $\mathsf{X}_{\textrm{Gauss}}^{n}$ have the same mean and the same covariance. Using the triangle inequality we can further upper-bound the RHS of (\ref{ProofLast-2}) yielding
\begin{IEEEeqnarray}{rCL}
&&D(f^{(n)})\nonumber\\
&&\le\left(c_1\sqrt{\mathrm{E}\left[\left\Vert\mathsf{X}^{n}\right\Vert^{2}\right]}+c_2\right)\inf\sqrt{\mathrm{E}\!\left[(\Vert\mathsf{X}_{\textrm{Gauss}}^{n}\Vert\!+\!\Vert\mathsf{X}^{n}\Vert)^{2}\right]}\nonumber\\ 
&&=\left(c_1\sqrt{\mathrm{E}\left[\left\Vert\mathsf{X}^{n}\right\Vert^{2}\right]}+c_2\right)\nonumber\\
&&\quad\times\inf\sqrt{\mathrm{E}\left[\Vert\mathsf{X}_{\textrm{Gauss}}^{n}\Vert^{2}+2\Vert\mathsf{X}_{\textrm{Gauss}}^{n}\Vert\Vert\mathsf{X}^{n}\Vert+\Vert\mathsf{X}^{n}\Vert^{2}\right]}\nonumber\\
&&\le\left(c_1\sqrt{\mathrm{E}\left[\left\Vert\mathsf{X}^{n}\right\Vert^{2}\right]}+c_2\right)\sqrt{4\mathrm{E}\left[\Vert\mathsf{X}^{n}\Vert^{2}\right]}\label{ProofLast-1}
\end{IEEEeqnarray}
where for (\ref{ProofLast-1}) we have applied the Cauchy-Schwarz inequality and the equality of the second moments of $\mathsf{X}^{n}$ and $\mathsf{X}_{\textrm{Gauss}}^{n}$. With (\ref{ProofLast-1}) we get
\begin{IEEEeqnarray}{rCL}
\frac{D(f^{(n)})}{n}&\le&\frac{2c_1 \mathrm{E}\left[\Vert\mathsf{X}^{n}\Vert^{2}\right]}{n}+\frac{2c_2 \sqrt{\mathrm{E}\left[\Vert\mathsf{X}^{n}\Vert^{2}\right]}}{n}\nonumber\\
&=&2c_1 \mathrm{E}\left[\Vert\mathsf{X}_k^{n}\Vert^{2}\right]+\frac{2c_2 \sqrt{\mathrm{E}\left[\Vert\mathsf{X}_k^{n}\Vert^{2}\right]}}{\sqrt{n}}\label{LastIneqArray1}\\
&\le&c<\infty \qquad\forall n\label{LastIneqArray2}
\end{IEEEeqnarray}
where in (\ref{LastIneqArray1}) $\mathsf{X}_k^{n}$ is the $k$-th element of the vector $\mathsf{X}^{n}$ and where we have used that $\{\mathsf{X}\}$ is stationary. Moreover, (\ref{LastIneqArray2}) holds as the second moments of $\mathsf{X}^{n}\sim f^{(n)}$ as well as $c_1$ and $c_2$ are finite and with $c$ being a uniform constant independent of $n$, which completes the proof.
\end{proof}

Theorem~\ref{TheoremBoundStationary} states that for a stationary $(c_1,c_2)$-regular process $\{\mathsf{X}\}$ the differential entropy rate is upper-bounded by the differential entropy rate of a Gaussian process with the same autocovariance function, such that $\{h(\mathsf{X}_{\textrm{Gauss}}^{n})-h(\mathsf{X}^{n})\}/n$ is not diverging for $n\rightarrow\infty$. For $n\rightarrow\infty$ we get with (\ref{LastIneqArray1}), cf. (\ref{NewLimDiffKL})\looseness-1
\begin{IEEEeqnarray}{rCL}
h'(\{\mathsf{X}_{\textrm{Gauss}}\})-h'(\{\mathsf{X}\})
&\le& 2c_1\mathrm{E}[\Vert\mathsf{X}^n_k\Vert^2]=2c_1(\mu^2+\sigma^2)\qquad
\end{IEEEeqnarray}
with $\mu$ and $\sigma^2$ being the mean and the variance of the elements of the process $\{\mathsf{X}\}$.\footnote{This holds in case $c_1$ and $c_2$ are independent of $n$, which is fulfilled by our definition of a $(c_1,c_2)$-regular process.} Differently stated, with (\ref{EntRateSpec}) it follows that the differential entropy rate $h'(\{\mathsf{X}\})$ is lower-bounded by
\begin{IEEEeqnarray}{rCL}
h'(\{\mathsf{X}\})&\ge \frac{1}{2} \int_{-\frac{1}{2}}^{\frac{1}{2}}\log(2\pi e S_{\mathsf{X}}(f))df-2c_1(\mu^2+\sigma^2).\ \label{LowerBoundRate}
\end{IEEEeqnarray}
As $h'(\{\mathsf{X}\})$ does not depend on the mean value of $\{\mathsf{X}\}$, we can set $\mu=0$.

\subsection{Example}
We consider the stationary random process $\{\mathsf{Y}\}$. Let $\mathsf{Y}^{n}$ be a vector containing a sequence of length $n$ of $\{\mathsf{Y}\}$. It is given by\looseness-1
\begin{IEEEeqnarray}{rCL}
\mathsf{Y}^{n}&=&\mathsf{H}^{n}\odot\mathsf{X}^{n}+\mathsf{Z}^{n}\label{ExampleModel}
\end{IEEEeqnarray}
where $\odot$ denotes the Hadamard product, i.e., the element wise product. Moreover, $\mathsf{H}^{n}$ and $\mathsf{X}^{n}$ are vectors containing sequences of length $n$ of the two zero-mean stationary Gaussian processes $\{\mathsf{H}\}$ and $\{\mathsf{X}\}$, respectively. These processes have the autocorrelation functions $r_{\mathsf{H}}(l)$ and $r_{\mathsf{X}}(l)$, cf. (\ref{DefAutocorr}). The corresponding power spectral densities are denoted by $S_{\mathsf{H}}(f)$ and $S_{\mathsf{X}}(f)$, cf. (\ref{Def_PSD}). Finally, the entries of the vector $\mathsf{Z}^{n}$ correspond to a sequence of length $n$ of the i.i.d.\ zero-mean Gaussian process $\{\mathsf{Z}\}$ and have variance $\sigma_{\mathsf{Z}}^{2}$. The processes $\{\mathsf{H}\}$, $\{\mathsf{X}\}$, and $\{\mathsf{Z}\}$ are mutually independent. The mean and the variance of the elements of the process $\{\mathsf{Y}\}$ are given by $\mu=0$ and $\sigma^{2}=r_{\mathsf{H}}(0)r_{\mathsf{X}}(0)+\sigma_{\mathsf{Z}}^{2}$. Moreover, the power spectral density $S_{\mathsf{Y}}(f)$ is given by
\begin{IEEEeqnarray}{rCL}
S_{\mathsf{Y}}(f)&=&(S_{\mathsf{H}}\star S_{\mathsf{X}})(f)+\sigma_{\mathsf{Z}}^{2},  \quad -\frac{1}{2}<f < \frac{1}{2}\label{PSD_YExamp}
\end{IEEEeqnarray}
where $\star$ denotes convolution.

The vector $\mathsf{Y}^{n}$ is not Gaussian distributed and no closed form for the probability density function is available. Moreover, no closed form solution for the differential entropy rate of the process $\{\mathsf{Y}\}$ is available. Thus, our aim is to lower-bound the differential entropy rate of $\{\mathsf{Y}\}$ using the approach described above. For this purpose we use the following proposition given by Polyanskiy and Wu.

\begin{Propos}\cite[Prop.~2]{PolyanskiyWu2015}\label{PropoRegConv}
Let $\mathsf{Y}^{n}=\mathsf{B}^{n}+\mathsf{Z}^{n}$ where $\mathsf{B}^{n}$ is independent of $\mathsf{Z}^{n}\sim\mathcal{N}(0,\sigma_{\mathsf{Z}}^{2}\mathbf{I}_{n})$, with $\mathbf{I}_{n}$ being the identity matrix of size $n\times n$. Then the density of $\mathsf{Y}^{n}$ is $(c_1,c_2)$-regular with $c_1=\frac{3\log e}{\sigma_{\mathsf{Z}}^{2}}$ and $c_2=\frac{4\log e}{\sigma_{\mathsf{Z}}^{2}}\mathrm{E}[\Vert\mathsf{B}^{n}\Vert]$. 
\end{Propos}

Using Proposition~\ref{PropoRegConv} it follows that $\mathsf{Y}^{n}$ in (\ref{ExampleModel}) is $(c_1,c_2)$-regular with $c_1=\frac{3\log e}{\sigma_{\mathsf{Z}^{2}}}$ and $c_2=\frac{4\log e}{\sigma_{\mathsf{Z}}^{2}} \sqrt{n r_{\mathsf{H}}(0)r_{\mathsf{X}}(0)}\ge \frac{4\log e}{\sigma_{\mathsf{Z}}^{2}}\mathrm{E}[\Vert \mathsf{H}^{n}\odot\mathsf{X}^{n} \Vert]$. However, as $c_2$ depends on $n$ we cannot directly apply (\ref{LowerBoundRate}) but still have to consider the second term on the RHS of (\ref{LastIneqArray1}) for bounding, yielding
\begin{IEEEeqnarray}{rCL}
h'(\{\mathsf{Y}\})&\ge&\frac{1}{2}\int_{-\frac{1}{2}}^{\frac{1}{2}}\!\log(2\pi e S_{\mathsf{Y}}(f))df-2c_1\sigma^2-\!\lim_{n\rightarrow\infty}\!\frac{2c_2 \sqrt{\sigma^2}}{\sqrt{n}}.\nonumber\\\label{ExtBound}
\end{IEEEeqnarray}
Introducing $c_1, c_2$, $\sigma^2$, and $S_{\mathsf{Y}}(f)$ yields
\begin{IEEEeqnarray}{rCL}
\hspace{-0.4cm}h'(\{\mathsf{Y}\})&\ge& \frac{1}{2} \int_{-\frac{1}{2}}^{\frac{1}{2}}\log(2\pi e ((S_{\mathsf{H}}\star S_{\mathsf{X}})(f)+\sigma_{\mathsf{Z}}^{2}))df\nonumber\\&&-6\log e \left(\frac{r_{\mathsf{H}}(0)r_{\mathsf{X}}(0)}{\sigma_{\mathsf{Z}}^{2}}+1\right)\nonumber\\
&&-\frac{8\log e}{\sigma_{\mathsf{Z}}^{2}}\sqrt{r_{\mathsf{H}}(0)r_{\mathsf{X}}(0)} \sqrt{r_{\mathsf{H}}(0)r_{\mathsf{X}}(0)+\sigma_{\mathsf{Z}}^{2}}.\label{ExamLB1}
\end{IEEEeqnarray}
The larger $\sigma_{\mathsf{Z}}^{2}$ becomes the more Gaussian the process $\{\mathsf{Y}\}$ gets and the closer the lower bound on the RHS of (\ref{ExamLB1}) comes to the differential entropy rate of a Gaussian process with the same power spectral density  which is given by $\frac{1}{2} \int_{-\frac{1}{2}}^{\frac{1}{2}}\log(2\pi e ((S_{\mathsf{H}}\star S_{\mathsf{X}})(f)+\sigma_{\mathsf{Z}}^{2}))df$. For large $\sigma_{\mathsf{Z}}^{2}$ the difference between the lower bound and the actual value of $h'(\mathsf{Y})$ converges to $6~\textrm{nats}$.

\section{Relation to the Hyperplane Conjecture}\label{Section_Hyperplane}
The methods applied to prove Theorem~\ref{TheoremBoundStationary} allow also to state the following theorem.
\begin{theorem}\label{TheoremPartialHyperPlane}
For any sequence of log-concave probability density functions $f^{(n)}$ on $\mathbb{R}^{n}$ for which the second moments are uniformly bounded and which is $(c_1,c_2)$-regular it holds that\looseness-1
\begin{IEEEeqnarray}{rCL}
\frac{D(f^{(n)})}{n}&\le& c\label{MainIneqPartialConjHyp}
\end{IEEEeqnarray}
with $c$ being a uniform constant independent of $n$.
\end{theorem}
The difference between Theorem~\ref{TheoremBoundStationary} and Theorem~\ref{TheoremPartialHyperPlane} is that in the latter one we do not assume that the sequence of probability density functions $f^{(n)}$ arises from a stationary process. Instead, we make the assumption that the sequence of $f^{(n)}$ has a uniform bound for all second moments. In this regard, note that in the proof of Theorem~\ref{TheoremBoundStationary} the assumption on stationarity is only used in (\ref{LastIneqArray1}). However, with the constraint on uniformly bounded second moments it similarly to (\ref{LastIneqArray1}) holds that\looseness-1
\begin{IEEEeqnarray}{rCL}
\frac{1}{n}D(f^{(n)})&\le&2c_1 \max_{k\in\{1,\hdots,n\}}\mathrm{E}\left[\Vert\mathsf{X}_k^{n}\Vert^{2}\right]\nonumber\\
&&+\frac{2c_2 \sqrt{\max_{k\in\{1,\hdots,n\}}\mathrm{E}\left[\Vert\mathsf{X}_k^{n}\Vert^{2}\right]}}{\sqrt{n}}\label{LastIneqArray1_HyperPartial}\\
&\le&c<\infty \qquad\forall n.\label{LastIneqArray2_HyperPartial}
\end{IEEEeqnarray}
Hence, with this slight change of the proof of Theorem~\ref{TheoremBoundStationary} Theorem~\ref{TheoremPartialHyperPlane} follows. 

Although Theorem~\ref{TheoremPartialHyperPlane} looks to some extent similar to the entropic form of the hyperplane conjecture (Conjecture~\ref{SlicingProb_Entrop}) there are some significant differences. On the one hand, Theorem~\ref{TheoremPartialHyperPlane} is limited to sequences of log-concave probability density functions which have a uniform bound for all second moments and which are $(c_1,c_2)$-regular. On the other hand, the constant $c$ in Theorem~\ref{TheoremPartialHyperPlane} is a uniform constant independent of $n$, however, it depends on the individual sequence of $f^{(n)}$. Differently, the constant $c$ in Conjecture~\ref{SlicingProb_Entrop} is universal and holds for any log-concave density $f^{(n)}$ on $\mathbb{R}^{n}$. Thus, by the methods applied in the present paper we do not obtain a proof of the hyperplane conjecture.

The hyperplane conjecture and, thus, Theorem~\ref{TheoremPartialHyperPlane} are formulated for probability density functions $f^{(n)}$ which are log-concave. The question is, which log-concave density functions are  $(c_1,c_2)$-regular such that Theorem~\ref{TheoremPartialHyperPlane} can be applied? First of all note that a regular density fulfilling (\ref{RegularityCond}) has infinite support on $\mathbb{R}^{n}$, i.e., $f^{(n)}$ is never zero. However, even for $f^{(n)}$ having infinite support neither does log-concavity imply $(c_1,c_2)$-regularity, nor does $(c_1,c_2)$-regularity imply log-concavity. Hence, the question remains, which log-concave densities $f^{(n)}$ with infinite support are $(c_1,c_2)$-regular? 

If we assume that $f^{(n)}$ is differentiable on $\mathbb{R}^{n}$, for any finite $x^{n}$ we can find a $c_1$ and $c_2$ such that (\ref{RegularityCond}) is fulfilled. In order to prove this consider that the derivative of $f^{(n)}$ is finite on any closed set. What remains is to check the case $\Vert x^{n} \Vert \rightarrow \infty$. In this regard, note that (\ref{RegularityCond}) implies that
\begin{IEEEeqnarray}{rCL}
|\log f^{(n)}(x^{n})-\log f^{(n)}(0)|&\le&\frac{c_1}{2}\Vert x^{n}\Vert^2 +c_2 \Vert x^{n}\Vert .\label{RegularityConseq}
\end{IEEEeqnarray}
Without loss of generality we can assume that the maximum of $f^{(n)}$ is attained at $x^{n}=0$, as a shift of the density $f^{(n)}$ has no impact on its log-concavity. As log-concavity of $f^{(n)}$ implies that $f^{(n)}$ is unimodal, it holds that $f^{(n)}(x^{n})\le f^{(n)}(0)$. Hence, (\ref{RegularityConseq}) yields\footnote{Interpreting the $\log$ as the natural logarithm.}
\begin{IEEEeqnarray}{rCL}
f^{(n)}(x^{n})&\ge&f^{(n)}(0) e^{-\left(\frac{c_1}{2}\Vert x^{n}\Vert^2 +c_2 \Vert x^{n}\Vert\right)}.
\end{IEEEeqnarray}
On the other hand, log-concave densities can be bounded from above using the following lemma given in \cite[Lemma~1]{CuleSamwoth2010}.
\begin{lemma}\label{BoundedDensityatInfinity}
Let $f^{(n)}$ be a log-concave density on $\mathbb{R}^{n}$. There exist $a>0$ and $b\in \mathbb{R}$ such that $ f^{(n)}(x^{n})\le e^{-a\Vert x^{n}\Vert+b}$ for all $x^{n}\in \mathbb{R}^{n}$. 
\end{lemma}
As the upper bound on log-concave densities in Lemma~\ref{BoundedDensityatInfinity} decays only with $e^{-a\Vert x^{n}\Vert}$ whereas $(c_1,c_2)$-regularity requires that the density decays not faster than with $e^{-\frac{c_1}{2}\Vert x^{n}\Vert^2}$, for every $(c_1,c_2)$-regular and log-concave density there exist constants $a$ and $b$ such that 
\begin{IEEEeqnarray}{rCL}
f^{(n)}(0) e^{-\left(\frac{c_1}{2}\Vert x^{n}\Vert^2 +c_2 \Vert x^{n}\Vert\right)}&\le& f^{(n)}(x^{n})\le e^{-a\Vert x^{n}\Vert+b}\quad\label{NewDecayBounds}
\end{IEEEeqnarray}
holds. I.e., densities being log-concave as well as $(c_1,c_2)$-regular have to be characterized by an exponential decay with a rate between the bounds given in (\ref{NewDecayBounds}). That means a log-concave density $f^{(n)}$ with infinite support on $\mathbb{R}^{n}$ which is differentiable and for which we can find constants $c_1$ and $c_2$ is $(c_1,c_2)$-regular in addition to being log-concave.

One example for a density falling into this class is the \emph{Logistic}-distribution whose density is given by
\begin{IEEEeqnarray}{rCL}
f(x)&=&\frac{e^x}{(1+e^x)^2}, \quad x\in \mathbb{R}.
\end{IEEEeqnarray}
Also the multivariate Gaussian density falls into this class.

In conclusion, although Theorem~\ref{TheoremPartialHyperPlane} is related to the entropic form of the hyperplane conjecture (both upper-bound $D(f^{(n)})/n$) and can be applied to any sequence of log-concave density functions $f^{(n)}$ which is $(c_1,c_2)$-regular and for which the second moments are uniformly bounded, there is the significant difference that the constant $c$ in Theorem~\ref{TheoremPartialHyperPlane} depends on the sequence of $f^{(n)}$ whereas it is a universal constant in the hyperplane conjecture.

\emph{Remark:} It has to be stated that we have only one concrete example of a sequence of $(c_1,c_2)$-regular log-concave probability density functions. Namely if the elements of $\mathsf{X}^{n}\sim f^{(n)}$ correspond to a sequence of a discrete-time zero-mean 
stationary Gaussian process $\{\mathsf{X}\}$ with autocovariance function $r(l)$ and corresponding power spectral density $S_{\mathsf{X}}(f)$, see (\ref{Def_PSD}), being supported in the entire interval $[-\frac{1}{2},\frac{1}{2}]$, then the sequence of densities $f^{(n)}$ is $(c_1,c_2)$-regular and log-concave. The condition on $S_{\mathsf{X}}(f)$ assures that the eigenvalues of the covariance matrices $\mathbf{R}_{\mathsf{X}}^{(n)}$ of $\mathsf{X}^{n}$ are bounded away from zero for any $n$. In this case we can find a $c_1$ such that (\ref{RegularityCond}) is fulfilled for all $n\in\mathbb{N}$. In this regard, consider that 
\begin{IEEEeqnarray}{rCL}
\Vert\nabla\log f^{(n)}(x^{n})\Vert&=&\Vert(\mathbf{R}^{(n)}_{\mathsf{X}})^{-1}x^{n}\Vert\nonumber\\
&\le&\Vert(\mathbf{R}^{(n)}_{\mathsf{X}})^{-1}\Vert\Vert x^{n}\Vert\label{Example_c1c2_1}\\ 
&=& \lambda_{\textrm{min}}^{-1}\Vert x^{n}\Vert.\label{Example_c1c2_2}
\end{IEEEeqnarray}
where (\ref{Example_c1c2_1}) follows with \cite[Theorem~5.6.2]{Horn} with $\Vert(\mathbf{R}^{(n)}_{\mathsf{X}})^{-1}\Vert$ being the spectral norm of $(\mathbf{R}^{(n)}_{\mathsf{X}})^{-1}$ \cite[p.~295]{Horn}. Thus, $\lambda_{\textrm{min}}$ is the minimal eigenvalue of $\mathbf{R}^{(n)}_{\mathsf{X}}$, which as stated above is bounded away from zero for all $n$.
However, for a multivariate Gaussian density it is obvious that  $D(f^{(n)})=0$ and, thus, the application of Theorem~\ref{TheoremPartialHyperPlane} is not necessary.

\section*{Acknowledgement}
The author thanks Heinrich Meyr for proofreading and discussing the paper and Martin Mittelbach and Tsachy Weissman for helpful discussions. Moreover, he thanks the reviewers for useful comments.

\bibliographystyle{IEEEtran}
\bibliography{IEEEabrv,Bib_mod_IEEE}

\end{document}